\documentclass[aps,pra,showpacs,reprint,footinbib,longbibliography,superscriptaddress]{revtex4-1}

\usepackage{etex}

\pdfoutput=1

\usepackage[greek,english]{babel}

\usepackage{amsmath,amsfonts,amssymb,amsthm,graphics,enumerate}
\usepackage{mathtools}

\usepackage[all]{xy}
\usepackage[dvipsnames]{xcolor}

\usepackage{soul}

\PassOptionsToPackage{hyphens}{url}

\usepackage[colorlinks=true,
            citecolor=blue,
            linkcolor=blue,
            urlcolor=black,
            allcolors=blue]
           {hyperref}

\usepackage{mdwlist}

\usepackage{enumitem}

\theoremstyle{definition}

\newcommand*{\mot}{~}

\newcommand{\uppi}{\text{\greektext{p}\latintext}}

\newcommand{\imun}{\textnormal{{i}}}  

\newtheorem{Theorem}{Theorem}

\newtheorem{Lemma}{Lemma} 
\newtheorem{Def}{Definition}

\def\Z{\mathbb{Z}}
\def\R{\mathbb{R}}

\newcommand{\ket}[1]{|{#1}\rangle}

\newcommand{\states}{\mathcal{S}}
\newcommand{\context}{M}
\newcommand{\inferable}{\mathcal{I}}
\newcommand{\joint}{\mathcal{J}}
\newcommand{\cheap}{\mathcal{O}}
\newcommand{\magic}{\mathcal{M}}
\newcommand{\scheme}{\magic_\cheap}

\newcommand\trick[1]{}

\begin{document}

\title{Contextuality as a resource for models of quantum computation on qubits}

\author{{\small $\text{Juan Bermejo-Vega}^{1,2}$, $\text{Nicolas Delfosse}^{3,4}$, $\text{Dan E. Browne}^5$, $\text{Cihan Okay}^6$, $\text{Robert Raussendorf}^7$}}

\affiliation{{\scriptsize  1: Dahlem Center for Complex Quantum Systems, Freie Universit{\"a}t Berlin, 14195 Berlin, Germany,\\
2: Max-Planck Institut f{\"u}r Quantum Optics, Theory Division, 85748 Garching, Germany,\\
3: Institute for Quantum Information and Matter, California Institute of Technology, Pasadena, California 91125, USA,\\ 
4: Department of Physics and Astronomy, University of California, Riverside, California 92521, USA,\\ 
5: Department of Physics and Astronomy, University College London, Gower Street, London WC1E 6BT, United Kingdom,\\  
6: Department of Mathematics, University of Western Ontario, London, Ontario N6A 5B7, Canada,\\
7: Department of Physics and Astronomy, University of British Columbia, Vancouver, British Columbia V6T 1Z1, Canada}}


\begin{abstract}
A central question in quantum computation is to identify the resources that are responsible for quantum speed-up. Quantum contextuality has been recently shown to be a resource for quantum computation with magic states for odd-prime dimensional qudits and two-dimensional systems with real wavefunctions. The phenomenon of state-independent contextuality poses a priori an obstruction to characterizing the case of regular qubits, the fundamental building block of quantum computation. Here, we establish contextuality of magic states as a necessary resource for a large class of quantum computation schemes on qubits. We illustrate our result with a concrete  scheme related to measurement-based quantum computation.
\vspace{-20pt}
\end{abstract}


\maketitle


The model of quantum computation by state injection (QCSI) \cite{BK}  is a leading paradigm of  fault-tolerance quantum computation. Therein, quantum gates are restricted to belong to a small set of classically simulable  gates, called Clifford gates \cite{Stab}, that admit simple fault-tolerant implementations \cite{Gottesman98TFQC}.  Universal quantum computation is achieved via injection of ``magic states'' \cite{BK}, which  are  the source of  quantum computational power of the model.

A central question in QCSI is to characterize the physical  properties that magic states need to exhibit in order to serve as universal resources. In this regard, quantum contextuality has recently been established as  a necessary resource for QCSI. This was first achieved for  \emph{quopit}  systems \cite{NegWi,Howard}, where the  local  Hilbert space dimension is an odd prime power,  and subsequently for local dimension two
with the case of  \emph{rebits}\mot\cite{ReWi}. In the latter,  the density matrix is constrained to be real at all times.

In this Letter we ask ``Can contextuality be established as a computational resource for QCSI on \emph{qubits}?''. This is not a straightforward extension of the quopit case because the multiqubit setting is complicated by the presence of state-independent contextuality among Pauli observables \cite{Merm2,Peres}. Consequently, every quantum state of $n\geq 2$ qubits is contextual with respect to  Pauli measurements, including the completely mixed one \cite{Howard}. It is thus clear that contextuality of  magic states  alone  cannot be a computational resource for every QCSI scheme on qubits. 

Yet, there exist qubit QCSI schemes for which contextuality of magic states \emph{is} a resource, and we identify them in this Letter.  
Specifically, we consider qubit QCSI schemes ${\cal{M}}_{\cal{O}}$ that satisfy the following two constraints: 
\begin{enumerate}[label=\textnormal{(C\arabic*)},leftmargin=21pt]
\item  Resource character. There exists a quantum state that does not exhibit contextuality with respect to  measurements available in ${\cal{M}}_{\cal{O}}$. \label{con:SIC}
\item Tomographic completeness. For any state $\rho$, the expectation value of any Pauli observable  can be inferred  via the allowed operations of the scheme.\label{con:TOM}
\end{enumerate}
The motivation for these constraints is the following. 

Condition \ref{con:SIC} constitutes  a minimal principle  that unifies,  simplifies and extends the quopit \cite{Howard} and rebit \cite{ReWi} settings.  While seemingly a weak constraint, it excludes the possibility of Mermin-type state-independent contextuality \cite{Merm2,Peres} among the available measurements (see Lemma~\ref{lemma:PhaseCovention} below). A priori, the absence of state-independent contextuality comes at a price. Namely, for any QCSI scheme  ${\cal{M}}_{\cal{O}}$ on $n\geq 2$ qubits, not all $n$-qubit Pauli observables can be measured. Thus, the question arises of whether this limits access to all $n$ qubits for measurement. As we show in this Letter, this does not have to be the case.

Addressing this question, we impose tomographic completeness as our technical condition for a “true” $n$-qubit QCSI scheme, cf.\mot\ref{con:TOM}. It means that any quantum state can be fully measured given sufficiently many copies. The rebit scheme\mot\cite{ReWi}, for example, does not satisfy this.

One of our results is that for any number $n$ of qubits there exists a QCSI scheme that satisfies both conditions \ref{con:SIC} and \ref{con:TOM}. The reason why both conditions can simultaneously hold lies in a fundamental distinction between observables that can be measured directly in a given qubit QCSI scheme from those that can only be inferred by measurement of other observables. The resulting qubit QCSI schemes resemble their quopit counterparts\mot\cite{NegWi,Howard} in the absence of state-independent contextuality, yet have full tomographic power for the multiqubit setting.

The main result of this Letter
is Theorem \ref{Res1}. It says that if the initial (magic) states of a qubit QCSI scheme are describable by a noncontextual hidden variable model (NCHVM) it becomes fundamentally impossible to implement a universal set of gates. We  highlight that  Theorem\mot\ref{Res1}  applies generally to \emph{any} scheme fulfilling the condition \ref{con:SIC}, including that of Ref.\mot\cite{ReWi}.

The condition \ref{con:SIC} plays a pivotal role in our analysis. It is clear that contextuality of the magic states can be a resource only if condition \ref{con:SIC} holds. In this Letter  
we establish the converse, namely that contextuality of the magic states is a resource for QCSI {\em{if}} condition\mot\ref{con:SIC} holds.  Therefore, condition\mot{}\ref{con:SIC} is the structural element that unifies the previously discussed quopit \cite{Howard} and rebit \cite{ReWi} case, and the qubit scenarios discussed here. Together, condition \ref{con:SIC} and Theorem \ref{Res1} characterize the contextuality types that are needed in quantum computation via state injection, showing that state-dependent contextuality with respect to Pauli observables is a universality resource.

As a final remark, we note that the measurements available in QCSI schemes  satisfying \ref{con:SIC} preserve positivity of suitable Wigner functions \cite{Raussendorf15QubitQCSI}.\\

\emph{Setting.}---An $n$-qubit Pauli observable $T_a$ is a  hermitian operator  with $\pm1$ eigenvalues of  form
\begin{equation}\label{eq:PauliOperator}
T_a:=\xi{(a)}Z(a_Z)X(a_X):=\xi{(a)}\bigotimes_{i=1}^n Z_i^{a_{Z_i}}\bigotimes_{j=1}^n X_j^{a_{X_j}},
\end{equation}
where $a:=(a_Z,a_X)$ is a $2n$-bit string and $\xi(a)$ is a phase. Pauli observables define an operator basis that we call $\mathcal{T}_n$.\\

A qubit scheme $\magic_\cheap$  of quantum computation via state injection (QCSI)  consists  of a resource $\magic$ of initial ``magic'' states and  3 kinds of allowed operations:
\begin{enumerate}
\item  Measurement of any Pauli observable in a set $\mathcal{O}$.
\item A group $G$ of ``free'' Clifford  gates  that preserve $\mathcal{O}$ via conjugation up to a  global phase.
\item Classical processing and feedforward.
\end{enumerate}
Adaptive circuits of operations 1-3 may be combined  with classical postprocessing in order to simulate  measurements of Pauli observables that are not in $\cheap$ (cf.\mot{}Fig.\mot\ref{fig:Inferable}).   We name the latter  ``\emph{inferable}'' and let $\inferable$ be the superset of $\cheap$ defined by them. Analogously, we let $\joint$ be the set of sets of compatible Pauli observables that can be inferred jointly, which  define the ``contexts'' of our computational model. As shown in Fig.\  \ref{fig:Inferable}, not every set of  compatible Pauli observables is necessarily in $\joint$.  Yet,  $A\in\inferable$ implies that $\{A\}\in \joint$. Furthermore, for any  pair of  observables  $\{A,B\}\in\joint$ and $\alpha\in\R$, the observables $AB$, $\alpha A$ can be inferred jointly by measuring $A,B$, since the eigenvalues of the latter determine those of the former. Hence, 
\begin{equation}\label{eq:InferabilityExample}
\{A,B\}\in\joint\quad \Rightarrow \quad \{A,B,AB,\alpha A\}\in\joint, \,\,\,\forall\alpha\in\R.
\end{equation}
Constraint \ref{con:TOM} holds if and only if $\mathcal{T}_n\subset \inferable$, i.e.,  if and only if the outcome distribution of any Pauli observable can be sampled via measurements in $\mathcal{O}$ and classical postprocessing.\\
\begin{figure}[h]
\begin{center}
\includegraphics[width=6.75cm]{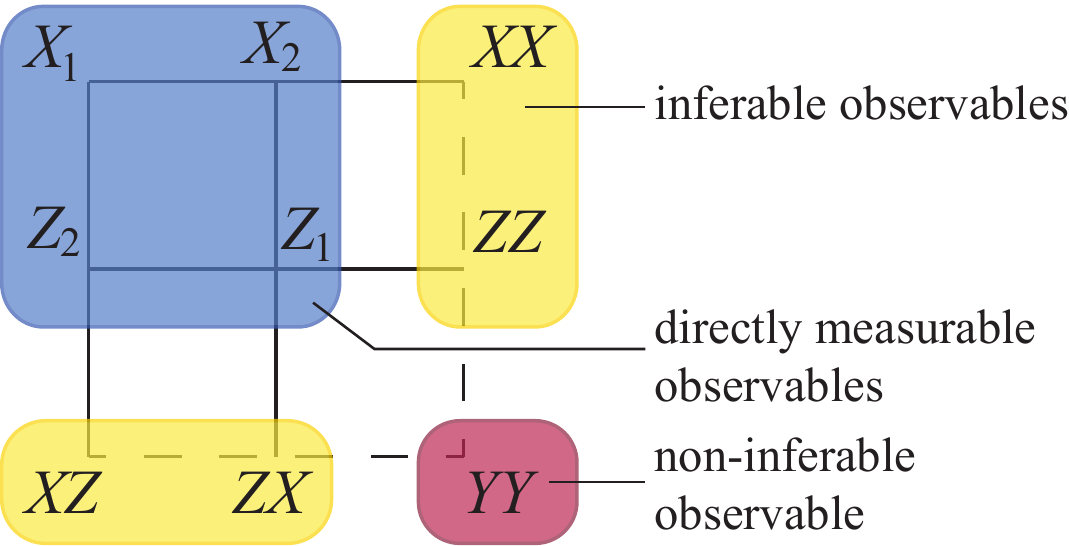}\caption[Inferring Observables]{\label{fig:Inferable} We consider an example scheme $\scheme$ on  two qubits with $\cheap=\{X_1,X_2,Z_1,Z_2\}$.  Straight lines connect maximal sets of jointly inferable observables.  Here, the correlator $X_1 X_2$ ($Z_1 Z_2$) is not in  $\cheap$ but  can be  inferred  by measuring $X_1,X_2$ ($Z_1,Z_2$) and multiplying their  outcomes. (This  scheme is reminiscent of  the syndrome measurement of subsystem codes \cite{Poulin}.) Yet, $X_1 X_2$ cannot be inferred jointly with $Z_1 Z_2$ because  a forbidden measurement of  $X_1,X_2,Z_1,Z_2$ would be required to reproduce all quantum  correlations, but after measuring, e.g.,  $Z_1$ and $Z_2$ to infer $Z_1Z_2$ the outcome statistics of $X_1X_2$ become uniformly random. Similarly, $X_1Z_2$ and $Z_1X_2$ can be separately inferred but not jointly. Further, $YY$ \emph{cannot} be inferred (observables in $\cheap$ cannot distinguish  its 	eigenstates).} 
\end{center}
\vspace{-5pt}
\end{figure}

\emph{Contextuality.}---Above, imposing \ref{con:SIC} means that there exists a quantum state $\rho$ whose measurement statistics can  be reproduced by a \emph{noncontextual hidden variable model} (NCHVM), which we introduce next.
{\begin{Def}\label{HVM1}
A \textnormal{NCHVM} $(\states, q_\rho,\Lambda)$ for the state $\rho$ with respect to a scheme $\scheme$  consists of a probability distribution $q_\rho$ over a set $\states$ of internal states and a set $\Lambda=\{\lambda_\nu\}_{\nu\in\states}$ of value  assignment functions  $\lambda_\nu:\inferable\rightarrow\{\pm 1\}$ that fulfill:
\begin{itemize*}
\item[(i)]  For any $\lambda_\nu\in\Lambda$ and $\context\in\joint$  the real numbers $\displaystyle\{\lambda_\nu(A)\}_{A\in M}$ are compatible eigenvalues: i.e.\,  there exists a quantum state $\ket{\psi}$ such that
\begin{equation}\label{eq:JointEigenvalue}
A\ket{\psi}=\lambda_\nu(A)\ket{\psi}, \quad\forall A\in \context.
\end{equation}
\item[(ii)] The distribution $q_\rho$ satisfies
\begin{equation}\label{eq:ExpectedValue}
\langle A \rangle_\rho = \mathrm{tr}(A \rho)=\sum_{\nu\in\mathcal{S}} \lambda_\nu(A)q_\rho(\nu), \quad \forall A\in\inferable
\end{equation}
\end{itemize*} The state $\rho$ is said to be ``contextual'' or to ``exhibit contextuality'' if no \textnormal{NCHVM} with respect to  $\scheme$ exists. 
\end{Def}
Qubit QCSI for which \emph{all} possible inputs exhibit contextuality are forbidden by \ref{con:SIC}. Specifically, in this Letter, 
$\cheap$ must be a strict subset of $\mathcal{T}_n$.\\

\emph{Main result}.---We  now establish contextuality as a resource for quantum computational universality for all qubit QCSI schemes  that fulfill \ref{con:SIC}.  Below, we call a scheme $\scheme$  \emph{universal} if  for any integer $n\geq 1$ and $V\in U(2^n)$  there exists a finite-size circuit of $\scheme$ operations that prepares the $n$-qubit state $V\ket{0}$ up to any positive trace-norm error.
\begin{Theorem}\label{Res1}
A qubit QCSI scheme $\scheme$ satisfying \ref{con:SIC} is  universal  for $n\geq 3$  qubits   only if its magic states exhibit contextuality.
\end{Theorem}
Theorem\mot\ref{Res1} applies even in the setting where the computation happens in an encoded subspace,  reproducing the  rebit results of Ref.\mot\cite{ReWi}. We provide a general proof of this fact in a companion paper\mot\cite{Raussendorf15QubitQCSI} and show it here in the encoding-free scenario under an additional  assumption, denoted ($\star$),  that every qubit must be measurable in at least two complementary Pauli bases. This requirement enforces $\scheme$ to exhibit the phenomenon of quantum complementarity and  simplifies our main argument while preserving its core structure.

The proof of theorem \ref{Res1} relies on a characterization of  noncontextual hidden variable models for  qubit QCSIs. We  make three key observations about such models.

First, by applying Def.\mot\ref{HVM1}.(i) to  $M:=\{A,B,AB,\alpha A\}\in\joint$ as in Eq.\mot{}(\ref{eq:InferabilityExample}), we derive two constraints
\begin{equation}\label{eq:LinearityOfAssignments}
\lambda_\nu(AB)=\lambda_\nu(A)\lambda_\nu(B),\quad  \lambda_\nu(\alpha A)= \alpha \lambda_\nu(A), 
\end{equation}
that any $\lambda_\nu\in\Lambda$  must fulfill  for any pair $\{A,B\}\in\joint, \alpha\in\mathbb{R}$. 

Second, we prove the following lemma. 
\begin{Lemma}\label{lemma:PhaseCovention} For any  QCSI scheme $\scheme$ fulfilling \ref{con:SIC} the phase $\xi{(a)}$ in Eq.\mot{}(\ref{eq:PauliOperator}) can be chosen w.l.o.g.\ so that 
\begin{equation}\label{eq:PauliAdditivity}
T_aT_b=T_{a+b}\quad  \textnormal{for any triple $\{T_a,T_b,T_aT_b\}\in\joint$}.
\end{equation}
\end{Lemma}
\begin{proof}
Let  $\xi$ be  given  and let  $\lambda_\nu$ be a consistent value assignment for the scheme  $\scheme$. W.l.o.g., we can  redefine $\mathcal{T}_n':=\{T_a':=\lambda_\nu(T_a)T_a,T_a\in \mathcal{T}_n\}$ and $\mathcal{O'}=\{T_a',T_a\in \cheap\}$ introducing a classical relabeling of measurement outcomes, without changing any quantum feature of the scheme. Using   $T_{a+b}=\pm T_{a}T_{b}$, we obtain
\begin{align}
&T_{a+b}'={\lambda_\nu(T_{a+b})}T_{a+b}={\lambda((\pm1) T_{a}T_{b})}(\pm 1)T_{a}T_{b}\notag\\
&\stackrel{(\ref{eq:LinearityOfAssignments})}{=}(\pm1)^2 {\lambda(T_{a}T_{b})} T_{a}T_{b}\stackrel{(\ref{eq:LinearityOfAssignments})}{=}{\lambda(T_{a})}T_{a}{\lambda(T_{b})}T_{b}=T_a'T_b'.
\:\:\:\:\,\,
\qedhere
\notag
\end{align}
\end{proof}
Last,  we observe that  for any $M\in\joint$, $\ket{\psi}$ as in  Eq.\mot{}(\ref{eq:JointEigenvalue}) and $T_b\in\mathcal{T}_n$, the state  $T_b\ket{\psi}$ is a joint eigenstate of $M$: 
\begin{equation}\label{eq:ActionEigenspaces}
(\gamma T_a) T_b\ket{\psi}=\left(\lambda_\nu(\gamma T_a)(-1)^{[a,b]}\right)T_b\ket{\psi},\,\forall\gamma T_a \in M,
\end{equation} where   $[a,b]:=a_X b_Z + a_Z  b_X\bmod 2$;
combined with Eq.\mot{}(\ref{eq:LinearityOfAssignments}), this induces a group action of  $\Z_2^{2n}$ on value assignments
\begin{equation}\label{Transl}
\lambda_\nu\quad\stackrel{u}{\rightarrow} \quad\lambda_{\nu +u}(T_a) := \lambda_\nu(T_a) (-1)^{[u,a]},\;\forall u \in V.
\end{equation}
With these tools, we  arrive at a powerful intermediate result, namely,  a method to construct NCHVMs that can simulate  qubit QCSIs on  noncontextual inputs.
\begin{Lemma}\label{Lemma:HVMSimul}   For any qubit scheme $\scheme$ fulfilling \ref{con:SIC} and   any quantum circuit $\mathcal{C}$ of $\scheme$ operations, if there exists a \textnormal{NCHVM}  $(\states, q_{\rho_\text{in}},\Lambda)$ for some given input state $\rho_\text{in}$, there then exists a \textnormal{NCHVM}  $(\states, q_{\rho_\text{out}},\Lambda)$ for the   output $\rho_\text{out}:=\mathcal{C}(\rho_\text{in})$. 
\end{Lemma}
Lemma\mot\ref{Lemma:HVMSimul} establishes that contextuality  cannot be freely generated in qubit QCSI. A surprising aspect of this fact is that it holds for  circuits that contain intermediate measurements. Intuitively,  unitary gates in $\mathcal{G}$ must induce an action on the set of noncontextual states since they preserve the set $\cheap$. However, the evolution of noncontextual states under measurement is far from intuitive since the latter can often prepare states that are inaccessible to gates\mot\cite{BermejoVega_12_GKTheorem}. 

  Lemma\mot\ref{Lemma:HVMSimul}  leads to a simple classical random-walk algorithm for sampling from the output distribution of all measurements in $\mathcal{C}$,  which is further efficient if oracles for sampling from  $q_{\rho_{\text{in}}}$ and computing any $\lambda_\nu\in\Lambda$ are given.  The random walk first samples a state $\nu_0\in\states$  from $q_{\rho_{\text{in}}}$ and, upon measurement of $T_{a_t}\in\cheap$  at time\mot$t$,  outputs  $\lambda_{\nu_t}(T_{a_t})$ given  $\nu_t$ and updates  $\nu_t \rightarrow \nu_t + a$ with $1/2$ probability. The correctness of this algorithm follows from Eq.\mot(\ref{eq:MeasurementUpdate}) below and is analyzed in detailed in Ref.\mot\cite{Raussendorf15QubitQCSI}.
\begin{proof}
We fix a  phase convention for $T_a$ so that Eq.\mot{}(\ref{eq:PauliAdditivity}) in Lemma\mot{}\ref{lemma:PhaseCovention} holds and introduce a simplified notation
\begin{equation}
\notag
\lambda_\nu(a):=\lambda_\nu(T_a), \quad \textnormal{where }T_a \in \inferable, a\in\Z_2^{2n}.
\end{equation} 
Because free unitaries preserve $\cheap$ they can be propagated out of $\mathcal{C}$ via conjugation. Hence, we  can w.l.o.g.\ assume that $\mathcal{C}$  consists only of  measurements. Our proof is by induction. At time $t=1$,  $\rho_1=\rho_\text{in}$ has an NCHVM by assumption. At any other time $t+1$,  given an NCHVM   $(\states, q_{\rho_t},\Lambda)$ for   the state $\rho_t$, we construct  an NCHVM $(\states, q_{\rho_{t+1}},\Lambda)$   for   $\rho_{t+1}$. Specifically, let $T_{a_t}\in {\cheap}$ be the observable measured at time $t$ with corresponding outcome $s_t\in\{\pm 1\}$,   ${s}_{\prec t}:=(s_1,\ldots,s_{t})$ be the string of prior measurement records, and $p(s_t|{s}_{\prec t})$ the conditional probability of measuring $s_t$;
we will now show that $\rho_{t+1}$ admits the hidden-variable representation
\begin{align}\label{eq:MeasurementUpdate}
q_{\rho_{t+1}}(\nu) &:= \frac{\delta_{{s_t},\lambda_\nu(a_t)}}{p(s_t|{s}_{\prec t})}   \frac{q_{\rho_{t}}(\nu) +  q_{\rho_{t}}(\nu +a_t)}{2},\vspace{2mm}
\end{align}
where    $p(s_t|{s}_{\prec t})$ can be predicted by the HVM, since  $2p(s_t|{s}_{\prec t})=\left\langle I+{s_t}T_{a_t}\right\rangle_{\rho_t} = \langle I\rangle_{\rho_t} +  {s_t}   \langle T_{{a}_t} \rangle_{\rho_t}$---which are known by the induction promise. Our goal is to show that $(\states, q_{\rho_{t+1}},\Lambda)$  predicts the expected value of any   $T_a \in\inferable$  measured at time $t+1$. For this, we derive a useful expression,
\begin{widetext}
\begin{align}
\langle T_{a} \rangle_{\rho_{t+1}}^{\mathrm{HVM}} = \displaystyle{{\sum_{\nu \in {\cal{S}}}} q_{\rho_{t+1}}(\nu) \lambda_\nu(a)} \vspace{1mm}
&\stackrel{(\ref{eq:MeasurementUpdate})}{=} \sum_{\nu \in {\cal{S}}} \tfrac{\delta_{{s_t},\lambda_\nu(a_t)}q_{\rho_t}(\nu)}{ 2p(s_t|{s}_{\prec t})}  \lambda_\nu(a)+\sum_{\nu \in {\cal{S}}}\tfrac{\delta_{{s_t},\lambda_\nu(a_t)}q_{\rho_t}(\nu+a_t)}{ 2p(s_t|{s}_{\prec t})} \lambda_\nu(a).\\
&\stackrel{(\ref{Transl})}{=}
 \sum_{\nu \in {\cal{S}}} \tfrac{\delta_{{s_t},\lambda_\nu(a_t)}q_{\rho_t}(\nu)}{ 2p(s_t|{s}_{\prec t})}  \lambda_\nu(a)+\tfrac{\delta_{{s_t},\lambda_\nu(a_t)}q_{\rho_t}(\nu)}{ 2p(s_t|{s}_{\prec t})} \lambda_\nu(a)(-1)^{[a, a_t]},\notag
\end{align}
\end{widetext}
which we evaluate on two cases:

\textbf{(A)} $T_{a}, T_{a_t}$ anticommute, hence, $[a,a_t]=1$. We get $\langle T_{a} \rangle_{\rho_{t+1}}^{\mathrm{HVM}}=0$, in agreement with quantum mechanics.

\textbf{(B)} $T_{a}, T_{a_t}$ commute. In this case $[a,a_t]=0$.  Using the identity $\delta_{s,\lambda}=(1+s\lambda)/2, s,\lambda\in\{\pm 1\}$, we obtain
\begin{align}
&\langle T_{a} \rangle_{\rho_{t+1}}^{\mathrm{HVM}} = \sum_{\nu \in {\cal{S}}} \frac{1+{s_t}\lambda_\nu({a}_t)}{2 p(s_t|{s}_{\prec t})} q_{\rho_t}(\nu)   \lambda_\nu(a) \notag\\
&\displaystyle\stackrel{(\ref{eq:LinearityOfAssignments})}{=} \frac{\sum_{\nu \in {\cal{S}}}q_{\rho_t}(\nu) \lambda_\nu(a)+{s_t}\sum_{\nu \in {\cal{S}}}q_{\rho_t}(\nu) \lambda_\nu(a+a_{t})}{2p(s_t|{s}_{\prec t})} \notag
\end{align}
Finally, by induction hypothesis, we arrive at
\begin{align}
\langle T_{a} \rangle_{\rho_{t+1}}^{\mathrm{HVM}} &=\displaystyle{\frac{\langle T_a\rangle_{\rho_t}  +{s_t}  \langle T_{a+a_t}\rangle_{\rho_t}}{2p(s_t|{s}_{\prec t})}} \stackrel{(\ref{eq:PauliAdditivity})}{=} \displaystyle{\frac{\text{tr} \left( \rho_t \frac{I+{s_t} T_{a_t}}{2} T_a \right)}{p(s_t|{s}_{\prec t})}  } \notag\\
&= \displaystyle{\text{tr} \left( \frac{\left[\frac{I+{s_t} T_{a_t}}{2} \rho_t \frac{I+{s_t} T_{a_t}}{2}\right]}{p(s_t|{s}_{\prec t})}  T_a \right)}=\mathrm{tr}\left(\rho_{t+1}T_a\right)\notag
\end{align}
which is again the quantum mechanical prediction. 
\end{proof}
Finally, we  prove our main result.
\begin{proof}[Proof of theorem \ref{Res1}]  We derive a contradiction by  assuming (A1) that $\scheme$ is universal  and (A2) that all magic states in $\magic$ are noncontextual. We first consider  the computation to be error-free and drop this assumption at the end. 

Recall that, by assumption ($\star$), two complementary Pauli observables, denoted $Z_i, X_i\in \cheap$ w.l.o.g., can be measured on any qubit.  By (A1), the scheme  $\scheme$ can prepare the encoded  GHZ state $\ket{\psi}$ that is uniquely stabilized by $X_1X_2X_3$, $-X_1Z_2Z_3$, $-Z_1X_2Z_3$, $ -Z_1Z_2X_3$. Furthermore, $\scheme$ can also  infer the value of any correlator of form $A_1A_2A_3$ with $ A_i\in\{X_i,Z_i\}$ (in particular,  $\ket{\psi}$'s stabilizers) by measuring  $A_1,A_2,A_3$ individually. Quantum mechanics predicts 
\begin{equation}\notag
\left\langle X_1X_2X_3 -  X_1Z_2Z_3 -  Z_1X_2Z_3 - Z_1Z_2X_3\right\rangle_\psi^{\mathrm{QM}}= 4.
\end{equation}
On the other hand,  by (A2) and Lemma \ref{Lemma:HVMSimul},  there exists an NCHVM for $\ket{\psi}$ with respect to all quadruples of form $\{A_1,A_2,A_3,A_1A_2A_3, A_i\in X_i,Z_i\}$.  Using constraint\mot(\ref{eq:LinearityOfAssignments}) for noncontextual value assignments,  we derive an inequality for the NCHVM's prediction  
\begin{equation}\notag
\left\langle X_1X_2X_3 -  X_1Z_2Z_3 -  Z_1X_2Z_3 -  Z_1Z_2X_3\right\rangle_\psi^{\mathrm{HVM}}\leq 2,
\end{equation} 
originally due to  Mermin \cite{Mermin90extreme_Q_entanglement}, which  contradicts quantum mechanics. Hence, either (A1) or (A2) must be false.

Last, our argument holds if  arbitrarily small  errors are present because  the HVM's prediction  deviates from the quantum mechanical one by a finite amount (larger than\mot2).
\end{proof}

\emph{A qubit QCSI scheme powered by contextuality.}-- Here we prove that for any number $n$ of qubits there exists a universal qubit QCSI scheme $\scheme$ that fulfills the conditions \ref{con:SIC} and \ref{con:TOM}. The $\cheap$ measurements available in this scheme are all single-qubit Pauli measurements, the group $\mathcal{G}$ contains all single-qubit Clifford gates, and the magic state is locally unitarily equivalent to a 2D cluster state. This family of examples demonstrates that the classification provided by our main result (Theorem \ref{Res1}) is not empty.

We now show that single-qubit Pauli measurents satisfy \ref{con:SIC} and \ref{con:TOM}. First, note that the value of any Pauli observable can be inferred by measuring its single-qubit tensor components; hence,  local QCSI  fulfills \ref{con:TOM}. Second, we show\mot\ref{con:SIC} is also met by giving a  NCHVM for the  mixed state $\rho=I/2^n$ with respect to single-qubit operations. The most general joint measurement in $\joint$ that we can implement with the latter is to measure $n$   single-qubit Paulis  $\sigma_{1},\ldots,\sigma_{n}$ on distinct qubits, which lets us  infer the value of any observable $\gamma\bigotimes_{i=1}^n \sigma_{i}^{\alpha_{i}}$ with $\alpha\in\Z_2^n,\gamma\in\R$. Hence, the function $\lambda_{0}(\bigotimes_{i=1}^n \sigma_{i}^{\alpha_{i}}):=1$, which is a joint eigenvalue of  $\{\bigotimes_{i=1}^n \sigma_{i}^{\alpha_{i}}:\alpha\in\Z_2^n\}$,  extends linearly to a value assignment fulfilling Def.\ \ref{HVM1}(i). Picking $\mathcal{T}_n=\{I,X,Y,Z\}^{\otimes n}$, we obtain an  NCHVM via (\ref{Transl})  with  value assignments $\lambda_b(T_a):=(-1)^{[a,b]},b\in\Z_2^{2n}$ wherein  $\rho$ corresponds to a probability distribution $q_\rho(b):=1/2^{2n}$: indeed, our HVM predicts $\langle \gamma T_a \rangle_\rho=\gamma$ for  $T_a=T_0= I$ and 0 otherwise, matching the quantum mechanical prediction---this can be checked by computing the average of  $\lambda_b(T_a)$ over $b$ in each case.

\begin{figure}[t]
\begin{center}
\centering
\includegraphics[width=1\linewidth]{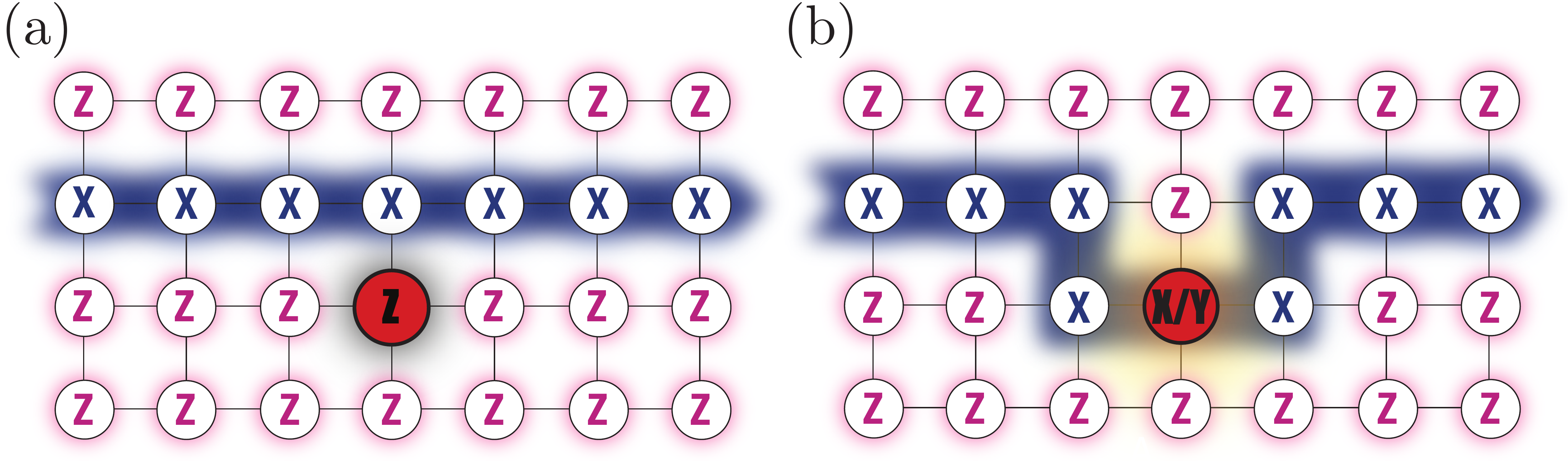}
\caption{\label{cluster}QCSI with modified cluster state $\ket{\Psi}$ and single-qubit  $X_i$, $Y_j$, $Z_k$ Pauli measurements: the $Z$ measurements are used to cut out of the plane a web corresponding to some layout of a quantum circuit, while the $X$ measurements drive the MBQC simulation of this circuit \cite{MBQC}. By ``re-routing'' a wire piece, one may choose between implementing and not implementing a non-Clifford gate. (a) Identity operation on the logical state space. (b) $X$ or $Y$ is measured adaptively to implement a logical $e^{-\imun\uppi/8\, Z}$ gate in MBQC \cite{MBQC}. }
\vspace{-10pt}
\end{center}
\end{figure}

Last, we present a family of magic states that  promote our local QCSI scheme to universality.  Unlike in  standard magic state distillation \cite{BK}, which relies on product magic states, our scheme has no entangling operations and requires entanglement to be present in the input to be universal. We show that a possibility is to use a modified cluster state $\ket{\Psi}$ that contains cells as in Fig.\ \ref{cluster} with ``red-site'' qubits that are locally rotated by a $T$ gate $e^{-\imun\uppi/8\, Z}$. Our approach is to use such state to simulate a universal scheme of measurement based quantum computation based on adaptive local measurements $\{Z,X,Y,X\pm Y/\sqrt{2}\}$ on a regular 2D cluster state\mot{}\cite{MBQC}.  Local Pauli measurements are available by assumption. Now, an on-site measurement of $X$ or $Y$ on one of the red qubits of  $\ket{\Psi}$ has the same effect as measuring $(X\pm Y)/\sqrt{2}$ on a cluster state.  To complete the simulation, it is enough to reroute the  measurement-based computation through a red site (this can be done with the available $X$ measurements\mot{}\cite{MBQC}) whenever a measurement of $(X\pm Y)/\sqrt{2}$ is needed. (See Fig.\ 2 for illustration.)

Note that an alternative resource state for one-qubit Pauli measurements is the so-called ``union-jack'' hypergraph state of Ref.\mot{}\cite{Miller1508.02695}.\\

\emph{Conclusion.}---In this Letter we investigated the role of contextuality in qubit QCSI and proved that it is a necessary resource for all such schemes that meet a simple  minimal condition: namely, that the allowed measurements do not exhibit state-independent contextuality. Our result applies if and only if contextuality emerges as a physical property possessed by  quantum states  (with respect to the measurements available in the computational model). We extended earlier results on odd-prime dimensional qudits \cite{NegWi,Howard} and rebits \cite{ReWi}, and thereby completed establishing contextuality as a resource in QCSI in arbitrary prime dimensions. We conjecture that this result generalizes to all composite dimensions\mot{}\cite{Gottesman98Fault_Tolerant_QC_HigherDimensions} (the composite odd case was recently covered after completion of this work\mot{}\cite{Delfosse16_equivalence}) and to algebraic extensions of QCSI models based on normalizer gates\mot\cite{VDNest_12_QFTs,BermejoVega_12_GKTheorem,BermejoLinVdN13_Infinite_Normalizers,BermejoLinVdN13_BlackBox_Normalizers,BermejoVegaZatloukal14Hypergroups}.  Further, we demonstrated the applicability of our result to a concrete qubit QCSI scheme that does not exhibit state independent contextuality while retaining tomographic completeness.
 
Finally, we refer to a companion paper\mot{}\cite{Raussendorf15QubitQCSI} where we investigate  the role of Wigner functions in qubit QCSI. There, we use Wigner functions to motivate the near-classical sector of the free operations in qubit QCSI, and relate their Wigner-function negativity to contextuality and hardness of classical simulation.  In comparison,  in this Letter,
constraint\mot{}\ref{con:SIC} completely removes the need to introduce Wigner functions, and leads us to the simplest and most general proof that contextuality can be a resource in  qubit QCSI that we are aware of. For this reason, we regard the establishing of condition \ref{con:SIC} as a fundamental structural insight of our Letter.\\

\emph{Acknowledgement.}---We thank David T.\ Stephen and the anonymous reviewers for comments on the manuscript. JBV acknowledges financial support by by Horizon 2020  (640800-AQuS-H2020-FETPROACT-2014). ND is funded by Institute for Quantum Information and Matter (IQIM), the National Science Foundation Physics Frontiers Center (PHY-1125565) and the Gordon and Betty Moore Foundation (GBMF-2644). CO is supported by the Natural Sciences and Engineering Research Council of Canada (NSERC). RR is funded by NSERC and the Canadian Institute for Advanced Research (CIFAR). RR is scholar of the CIFAR Quantum Information Science program.

\bibliography{database}

\newpage

\appendix

\end{document}